\documentclass[11pt]{article}

\usepackage[backend=biber,style=alphabetic,maxbibnames=100]{biblatex}
\usepackage{fullpage}
\usepackage[utf8]{inputenc}
\usepackage[dvipsnames,svgnames,x11names,hyperref]{xcolor}
\usepackage[colorlinks]{hyperref}
\usepackage{enumerate}
\usepackage{nicefrac}
\usepackage{graphicx}
\usepackage[textsize=scriptsize]{todonotes}
\usepackage{algorithm}
\usepackage{amsmath}
\usepackage{mathtools}
\usepackage{authblk}
\usepackage{amssymb}
\usepackage{amsthm}
\usepackage{esvect}
\usepackage[normalem]{ulem}
\usepackage{braket}
\usepackage{pgfplots}
\usepackage{multicol}
\usepackage[final]{microtype}

\pgfplotsset{compat=1.6}

\bibliography{bibliography}

\hypersetup{
	colorlinks = true,
	citecolor = teal,
	urlcolor = purple,
}

\newtheorem{theorem}{Theorem}
\newtheorem{lemma}[theorem]{Lemma}

\newtheorem{proposition}[theorem]{Proposition}

\DeclareMathOperator{\tw}{\mathbf{tw}}
\DeclareMathOperator{\TWR}{TWR}
\DeclareMathOperator{\TW}{TW}
\DeclareMathOperator{\be}{H}
\DeclareMathOperator{\poly}{poly}

\let\originalleft\left
\let\originalright\right
\renewcommand{\left}{\mathopen{}\mathclose\bgroup\originalleft}
\renewcommand{\right}{\aftergroup\egroup\originalright}

\DeclareFontFamily{U}{mathx}{\hyphenchar\font45}
\DeclareFontShape{U}{mathx}{m}{n}{
      <5> <6> <7> <8> <9> <10>
      <10.95> <12> <14.4> <17.28> <20.74> <24.88>
      mathx10
      }{}
\DeclareSymbolFont{mathx}{U}{mathx}{m}{n}
\DeclareMathSymbol{\bigtimes}{1}{mathx}{"91}

\makeatletter
\newcommand\dotline{\@ifnextchar[
  \answerlinetowidth\answerlinetoeol}
\newcommand\answerlinetowidth[1][0pt]{\hbox to #1{\leaders\hbox to \answerdotsep{\hss.\hss}\hfill}}
\newcommand\answerlinetoeol{\leaders\hbox to \answerdotsep{\hss.\hss}\hfill\strut}
\newcommand\answerdotsep{0.187cm}
\makeatother

\makeatletter
\newenvironment{breakablealgorithm}
  {
     \refstepcounter{algorithm}
     \hrule height.8pt depth0pt \kern2pt
     \renewcommand{\caption}[2][\relax]{
       {\raggedright\textbf{\fname@algorithm~\thealgorithm} ##2\par}%
       \ifx\relax##1\relax 
         \addcontentsline{loa}{algorithm}{\protect\numberline{\thealgorithm}##2}%
       \else 
         \addcontentsline{loa}{algorithm}{\protect\numberline{\thealgorithm}##1}%
       \fi
       \kern2pt\hrule\kern2pt
     }
  }{
     \kern2pt\hrule\relax
  }
\makeatother

\makeatletter
\newtheorem*{rep@theorem}{\rep@title}
\newcommand{\newreptheorem}[2]{%
\newenvironment{rep#1}[1]{%
 \def\rep@title{#2 \ref{##1}}%
 \begin{rep@theorem}}%
 {\end{rep@theorem}}}
\makeatother
\newreptheorem{theorem}{Theorem}

\title{Quantum speedups for treewidth}
\author{Vladislavs Kļevickis}
\author{Krišjānis Prūsis}
\author{Jevgēnijs Vihrovs}

\affil{Centre for Quantum Computer Science, Faculty of Computing,\authorcr University of Latvia, Raiņa 19, Riga, Latvia, LV-1586}

\date{}

\begin{document}

\maketitle

\begin{abstract}
In this paper, we study quantum algorithms for computing the exact value of the treewidth of a graph.
Our algorithms are based on the classical algorithm by Fomin and Villanger (Combinatorica 32, 2012) that uses $O(2.616^n)$ time and polynomial space.
We show three quantum algorithms with the following complexity, using QRAM in both exponential space algorithms:
\begin{itemize}
    \item $O(1.618^n)$ time and polynomial space;
    \item $O(1.554^n)$ time and $O(1.452^n)$  space;
    \item $O(1.538^n)$ time and space.
\end{itemize}
In contrast, the fastest known classical algorithm for treewidth uses $O(1.755^n)$ time and space.
The first two speed-ups are obtained in a fairly straightforward way. 
The first version uses additionally only Grover's search and provides a quadratic speedup.
The second speedup is more time-efficient and uses both Grover's search and the quantum exponential dynamic programming by Ambainis et al.~(SODA '19).
The third version uses the specific properties of the classical algorithm and treewidth, with a modified version of the quantum dynamic programming on the hypercube.
Lastly, as a small side result, we also give a new classical time-space tradeoff for computing treewidth in $O^*(2^n)$ time and $O^*(\sqrt{2^n})$ space.
\end{abstract}

\section{Introduction} \label{sec:intro}

For many NP-complete problems, the exact solution can be found much faster than a brute-force search over the possible solutions; it is not so rare that the best currently known algorithms are exponential \cite{FK10}.
Perhaps one of the most famous examples is the travelling salesman problem, where a naive brute-force requires $O^*(n!)$ computational time, but a dynamic programming algorithm solves it exactly only in $O^*(2^n)$ time \cite{Bel62, HK62}.
Such algorithms are studied also because they can reveal much about the mathematical structure of the problem and because sometimes in practice they can be more efficient than subexponential algorithms with a large constant factor in their complexity.

With the advent of quantum computing, it is curious how quantum procedures can be used to speed up such algorithms.
A clear example is illustrated by the SAT problem: while iterating over all possible assignments to the Boolean formula on $n$ variables gives $O^*(2^n)$ time, Grover's search \cite{Grover96} can speed this up quadratically, resulting in $O^*(\sqrt{2^n})$ time.
Grover's search can also speed up exponential dynamic programming: recently Ambainis et al.~\cite{ABIKPV19} have shown how to apply Grover's search recursively together with classical precalculation to speed up the $O^*(2^n)$ dynamic programming introduced by Bellman, Held and Karp \cite{Bel62, HK62} to a $O(1.817^n)$ quantum algorithm.
For some problems like the travelling salesman problem and minimum set cover, the authors also gave a more efficient $O(1.728^n)$ time quantum algorithm by combining Grover's search with both divide \& conquer and dynamic programming techniques.
Their approach has been subsequently applied to find a speedup for more NP-complete problems, including graph coloring \cite{SM20}, minimum Steiner tree \cite{MIKL20} and optimal OBDD ordering \cite{Tan20}.

In this paper, we focus on the NP-complete problem of finding the treewidth of a graph.
Informally, the treewidth is a value that describes how close the graph is to a tree; for example, the treewidth is $1$ when the graph is a tree, while the treewidth of a complete graph on $n$ vertices is $n-1$.
This quantity is prominently used in parameterized algorithms, as many problems are efficiently solvable when treewidth is small, such as vertex cover, independent set, dominating set, Hamiltonian cycle, graph coloring, etc.~\cite{AP89}.
The applications of treewidth, both theoretical and practical, are numerous, see \cite{Bod05} for a survey.
If the treewidth is at most $k$, it can be computed exactly in $O(n^{k+2})$ time \cite{ACP87}; $2$-approximated in parameterized linear time $2^{O(k)} n$ \cite{Kor21}; $O(\sqrt{\log k})$-approximated in polynomial time \cite{FHL08}; $k$-approximated in $O(k^7 n \log n)$ time \cite{FLSPW18}.

As for exact exponential time treewidth algorithms, both currently most time and space efficient algorithms were proposed by Fomin and Villanger in \cite{FV12}: the first uses $O(1.755^n)$ time and space and the second requires $O(2.616^n)$ time and polynomial space.
The crucial ingredient of these algorithms is a combinatorial lemma that upper bounds the number of connected subsets with fixed neighborhood size (Lemma \ref{thm:comblem}), as well as gives an algorithm that lists such sets.

Our main motivation for tackling these algorithms is that although the $O(1.817^n)$ quantum algorithm from \cite{ABIKPV19} is applicable to treewidth, it is still less efficient than Fomin's and Villanger's.
In this paper we show that their techniques are also amenable to quantum search procedures.
In particular, we focus on their polynomial space algorithm.
This algorithm has two nested procedures: the first procedure uses Lemma \ref{thm:comblem} to search through specific subsets of vertices $S$ to be fixed as a bag of the tree decomposition; the second procedure finds the optimal width of the tree decomposition with $S$ as its bag.

We find that Grover's search can be applied to the listing procedure of Lemma \ref{thm:comblem}, thus speeding up the first procedure quadratically.
For the second procedure, classically one can use either the $O^*(2^n)$ time and space dynamic programming algorithm or the $O^*(4^n)$ time and polynomial space divide \& conquer algorithm (Fomin and Villanger use the latter), which both were introduced in \cite{BFKAKT12}.
The divide \& conquer algorithm we can also speed up using Grover's search.
Thus, we obtain a quadratic speedup for the polynomial space algorithm:
\begin{reptheorem}{thm:qtwdq}
There is a bounded-error quantum algorithm that finds the exact treewidth of a graph on $n$ vertices in $O(1.61713^n)$ time and polynomial space.
\end{reptheorem}
\noindent Next, using the fact that the $O^*(2^n)$ dynamic programming algorithm can be sped up to an $O^*(1.817^n)$ quantum algorithm together with the quadratic speedup of Lemma \ref{thm:comblem}, we obtain our second quantum algorithm:
\begin{reptheorem}{thm:qtwdp}
Assuming the QRAM data structure, there is a bounded-error quantum algorithm that finds the exact treewidth of a graph on $n$ vertices in $O(1.55374^n)$ time and $O(1.45195^n)$ space.
\end{reptheorem}

The last theorem suggests a possibility for an even more efficient algorithm by trading some space for time.
We achieve this by proving a treewidth property which essentially states that we can precalculate some values of dynamic programming for the original graph, and reuse these values in the dynamic programming for its subgraphs (Lemma \ref{thm:global}). 
This allows us a global precalculation, which can be used in the second procedure of the treewidth algorithm.
To do that, we have to modify the $O(1.817^n)$ algorithm of \cite{ABIKPV19}.
We refer to it as the asymmetric quantum exponential dynamic programming.
This gives us the following algorithm:
\begin{reptheorem}{thm:main}
Assuming the QRAM data structure, there is a bounded-error quantum algorithm that finds the exact treewidth of a graph on $n$ vertices in $O(1.53793^n)$ time and space.
\end{reptheorem}

Lastly, we observe that replacing the $O^*(4^n)$ divide \& conquer algorithm in the classical $O(2.616^n)$ polynomial space algorithm by the $O^*(2^n)$ dynamic programming only lowers the time complexity to $O^*(2^n)$.
However, the interesting consequence is that the space requirement becomes only $O^*(\sqrt{2^n})$.
Hence, we obtain a \emph{classical} time-space tradeoff:
\begin{reptheorem}{thm:tradeoff}
The treewidth of a graph with $n$ vertices can be computed in $O^*(2^n)$ time and $O^*\left(\sqrt{2^n}\right)$ space.
\end{reptheorem}
\noindent Time-wise, this is more efficient than the $O(2.616^n)$ time polynomial space algorithm, and space-wise, this is more efficient than the $O(1.755^n)$ time and space algorithm.
It also fully subsumes the time-space tradeoffs for permutation problems proposed in \cite{KP10} applied to treewidth.

\section{Preliminaries} \label{sec:prelim}

We denote the set of integers from $1$ to $n$ by $[n]$.
For a set $S$, denote the set of all its subsets by $2^S$.
We call a permutation of a set of vertices $S \subseteq V$ a bijection $\pi : S \to [|S|]$.
We denote the set of permutations of $S$ by $\Pi(S)$.
For a permutation $\pi \in \Pi(S)$, let $\pi_{<v} = \{w \mid \pi(w) < \pi(v)\}$ and  $\pi_{>v} = \{w \mid \pi(w) > \pi(v)\}$.

We write $O(f(n)) = \poly(n)$ if $f(n) = O(n^c)$ for some constant $c$.
Also let $O(\poly(n) f(m)) = O^*(f(m))$.
This is useful since our subprocedures will often have some running time $f(m)$ times some function that depends on the size of the input graph $G$ on $n$ vertices.
In this paper, we are primarily concerned with the exponential complexity of the algorithms, hence, we are interested in the $f(m)$ value of an $O^*(f(m))$ complexity.

\paragraph{Graph notation.}

For a graph $G=(V,E)$ and a subset of vertices $S \subseteq V$, denote $G[S]$ as the graph induced in $G$ on $S$.
For a subset of vertices $S \subseteq V$, let $N(S) = \{ v \in V - S \mid u \in S, \{u,v\} \in E\}$ be its \emph{neighborhood}.
We call a subset $S \subseteq V$ connected if $G[S]$ is connected, and $C \subseteq V$ a clique if $G[C]$ is a complete graph.
Later on we also mention the notions of \emph{potential maximum cliques} and \emph{minimal separators}, which are specific subsets of $V$, but we don't rely on them; for their definitions and properties, see e.g.~\cite{FV12}.

\paragraph{Treewidth.}

A \emph{tree decomposition} of a graph $G=(V,E)$ is a pair $(X,T)$, where $T = (V_T,E_T)$ and $X = \{\chi_i \mid i \in V_T\} \subseteq 2^V$ such that:
\begin{itemize}
    \item $\bigcup_{\chi \in X} \chi = V$;
    \item for each edge $\{u,v\} \in E$, there exists $\chi \in X$ such that $u, v \in \chi$;
    \item for any vertex $v \in V$ in $G$, the set of vertices $\{\chi \mid v \in \chi\}$ forms a connected subtree of $T$.
\end{itemize}
We call the subsets $\chi \in X$ \emph{bags} and the vertices of $T$ \emph{nodes}.
The \emph{width} of $T$ is defined as the minimum size of $\chi \in X$ minus $1$.
The \emph{treewidth} of $G$ is defined as the minimum width of a tree decomposition of $G$ and we denote it by $\tw(G)$.

We also consider optimal tree decompositions given that some subset $\chi \in V$ is a bag of the tree.
We denote the smallest width of a tree decomposition of $G$ among those that contain $\chi$ as a fixed bag by $\tw(G,\chi)$.

\paragraph{Approximations.}

For the binomial coefficients, we use the following well-known approximation:
\begin{theorem}[Entropy approximation] \label{thm:entropy}
For any $k \in [0,1]$, we have
\[\binom{n}{k} \leq 2^{\be\left(\frac{k}{n}\right) \cdot n},\]
where $\be(\epsilon) = -(\epsilon \log_2(\epsilon) + (1-\epsilon) \log_2(1-\epsilon))$ is the binary entropy function.
\end{theorem}

\paragraph{Quantum subroutines.}

Our algorithms use a well-known variation of Grover's search, quantum minimum finding:
\begin{theorem}[Theorem 1 in \cite{DH96}] \label{thm:qmf}
Let $\mathcal A : N \to [n]$ be an exact quantum algorithm with running time $T$.
Then there is a bounded-error quantum algorithm that computes $\min_{i \in [N]} \mathcal A(i)$ in $O^*(T \sqrt{N})$ time.
\end{theorem}

Two of our algorithms use the QRAM data structure \cite{GLM08}.
This structure stores $N$ memory entries and, given a superposition of memory indices together with an empty data register $\sum_{i \in [N]} \alpha \ket{i} \ket{\mathbf{0}}$, it produces the state $\sum_{i \in [N]} \alpha \ket{i} \ket{\text{data}_i}$ in $O(\log N)$ time.
In our algorithms, $N$ will always be exponential in $n$, which means that a QRAM operation is going to be polynomial in $n$.
Thus, this factor will not affect the exponential complexity, which we are interested in.

In our algorithms, we will often have a quantum algorithm that takes exact subprocedures (like in Theorem \ref{thm:qmf}), and give it bounded-error subprocedures.
Since we always going to take $O(\exp(n))$ number of inputs, this issue can be easily solved by repeating the subprocedures $\poly(n)$ times to boost the probability of correct answer to $1-O(1/\exp(n))$: it can be then shown that the branch in which all the procedures have correct answers has constant amplitude.
The final bounded-error algorithm incurs only a polynomial factor, and does not affect the exponential complexity.
We also note that on a deeper perspective, all our quantum subroutines are based on the primitive of Grover's search \cite{Grover96}; an implementation of Grover's search with bounded-error inputs that does not incur additional factors in the complexity has been shown in \cite{HMDw03}.

We also are going to encounter an issue that sometimes we have some real parameter $\alpha \in [0,1]$ and we are examining $\binom{n}{\alpha n}$.
Since $\alpha n$ is not integer, this value is not defined; however, we can take this to be any value between $\binom{n}{\lfloor \alpha n \rfloor}$ or $\binom{n}{\lceil \alpha n \rceil}$, as they differ only by a factor of $n$.
Thus, this does not produce an issue for the exponential complexity analysis.
Henceforward we abuse the notation and simply write $\binom{n}{\alpha n}$.

\section{Combinatorial lemma} \label{sec:comblem}

In this section we describe how the main combinatorial lemma of \cite{FV12} can be sped up quantumly qudratically using Grover's search.

\begin{lemma}[Lemmas 3.1.~and 3.2.~in \cite{FV12}] \label{thm:comblem}
Let $G=(V,E)$ be a graph.
For every $v \in V$ and $b, f \geq 0$, the number of connected subsets $B \subseteq V$ such that
\begin{enumerate}
    \item $v \in B$,
    \item $|B| = b+1$, and
    \item $|N(B)|=f$
\end{enumerate}
is at most $\binom{b+f}{b}$.
There also exists an algorithm that lists all such sets in $O^*(\binom{b+f}{b})$ time and polynomial space.
\end{lemma}

Informally, this lemma is used in the treewidth algorithm to search for a set, such that, if fixed as a bag of the tree decomposition, the remaining graph breaks down into connected components of bounded size; then, the optimal width of the tree decomposition with this bag fixed can be solved using algorithms from Section \ref{sec:fixbag}.
The lemma gives an upper bound on the number of sets to consider.

Their proof of this lemma essentially gives a branching algorithm that splits the problem into several problems of the same type, and solves them recursively.
The idea for applying Grover's search to such a branching algorithm is simple.
The algorithm that generates all sets can be turned into a procedure that, given a number $i \in [\binom{b+f}{b}]$ of the set we need to generate, generates this set in polynomial time.
Then, we can run Grover's search over all integers in $\binom{b+f}{b}$ on this procedure.
This was formalized by Shimizu and Mori:

\begin{lemma}[Lemma 4 in \cite{SM20}]
Let $P$ be a decision problem with parameters $n_1, \ldots, n_\ell$.
Suppose that there is a branching rule $b(P)$ that reduces $P$ to $m_{b(P)}$ problems $P_1, \ldots, P_{m_{b(P)}}$ of the same class.
Here, $P_i$ has parameters $f_j^{b(P),i}(n_j)$ for $j\in[\ell]$, where $f_j^{b(P),i} \leq n_j$.
At least one of the parameters of $P_i$ must be strictly smaller than the corresponding parameter of $P$.
The solution for $P$ is equal to the minimum of the solutions for $P_1$, $\ldots$, $P_{m_b(P)}$.

Let $U(n_1,\ldots,n_\ell)$ be an upper bound on the number of leaves in the computational tree.
Assume that the running time of computing $b(P)$, $P_i$, $f_j^{b(P),i}$ and $U(n_1,\ldots,n_\ell)$ is polynomial w.r.t.~$n_1$, $\ldots$, $n_{\ell}$.
Suppose that
$U(n_1,\ldots,n_\ell) \geq \sum_{i=1}^{m_b(P)} U(f_1^{b(P),i}(n_1),\ldots,f_\ell^{b(P),i}(n_\ell))$.
Also suppose that $T$ is the running time for the computation at each of the leaves in the computational tree.
Then there is a bounded-error quantum algorithm that computes $P$ and has running time $\poly(n_1,\ldots,n_\ell)\sqrt{U(n_1,\ldots,n_\ell)}T$.
\end{lemma}

We apply this lemma to Lemma \ref{thm:comblem}:
\begin{lemma} \label{thm:quantum-comb}
Let $G=(V,E)$ be a graph, and $\mathcal A : 2^V \to \mathbb [n]$ be an exact quantum algorithm with running time $T$.
For every $v \in V$ and $b, f \geq 0$, let $\mathcal B_{v,b,f}$ be the set of connected subsets $B \subseteq V$ satisfying the conditions of Lemma \ref{thm:comblem}.
Then there is a bounded-error quantum algorithm that computes $\min_{B \in \mathcal B_{v,b,f}} \mathcal A(B)$ in time
$$O^*\left(\sqrt{\binom{b+f}{b}}\right).$$
\end{lemma}

\begin{proof}
According to the proof of Lemma \ref{thm:comblem} in \cite{FV12}, we have that
\begin{itemize}
    \item $\ell=2$, $n_1 = b$, $n_2 = f$.
    \item $b(P)$ splits the problem into $m_{b(P)} = f+b$ problems.
    \item $P_i$ has parameters $f_1^{b(P),i}(b) = b-1$ and $f_2^{b(P),i}(f) = f-i+1$.
    \item $U(b,f) = \binom{b+f}{f}$.
    \item $\sum_{i=1}^{m_b(P)} U\left(f_1^{b(P),i}(b),f_2^{b(P),i}(f)\right) = \sum_{i=1}^{f+b} \binom{f+b-i}{b-1} = \sum_{i=0}^{f+b-1} \binom{f+b-1-i}{b-1} = \binom{b+f}{f} = U(b,f)$.
    \item Computing $b(P)$, $f_1^{b(P),i}$, $f_2^{b(P),i}$ and $U(b,f)$ takes time polynomial in $b$ and $f$; computing $P_i$ involves contracting two vertices in the graph and can be done in $\poly(n)$ time. \qedhere
\end{itemize}
\end{proof}

\section{Fixed bag treewidth algorithms} \label{sec:fixbag}

In this section we describe algorithms that calculate the optimal treewidth of a graph with the condition that a subset of its vertices is fixed as a bag of the tree decomposition.
We then show ways to speed them up quantumly.
Both approaches were given by Bodlaender et al.~\cite{BFKAKT12}.

\subsection{Treewidth as a linear ordering}

Both of these algorithms use the fact that treewidth can be seen as a graph linear ordering problem.
For a detailed description, see Section 2.2 of \cite{BFKAKT12}, from where we also borrow a lot of notation.
We will also use the properties of this formulation in our improved quantum algorithm.

A linear ordering of a graph $G=(V,E)$ is a permutation $\pi \in \Pi(V)$.
The task of a linear ordering problem is finding $\min_{\pi \in \Pi(V)} f(\pi)$, for some known function $f$.

For two vertices $v, w \in V$, define a predicate $P_\pi(v,w)$ to be true iff there is a path from $v$ to $w$ in $G$ such that all internal vertices in that path are before $v$ and $w$ in $\pi$.
Then define $R_\pi(v)$ to be the number of vertices $w$ such that $\pi(w) > \pi(v)$ and $P_\pi(v,w)$ holds.
The following proposition gives a description of treewidth as a linear ordering problem:
\begin{proposition}[Proposition 3 in \cite{BFKAKT12}]
Let $G=(V,E)$ be a graph, and $k$ a non-negative integer.
The treewidth of $G$ is at most $k$ iff there is a linear ordering $\pi$ of $G$ such that for each $v \in V$, we have $R_\pi(v) \leq k$.
\end{proposition}

For a set of vertices $S \subseteq V$ and a vertex $v \notin S$, define
\[Q_G(S,v) = |\{ w \in V - S - \{v\} \mid \text{$v$ and $w$ are connected by a path in $G[S \cup \{v,w\}]$}\}|.\]
Note that $R_\pi(v) = |Q_G(\pi_{<v},v)|$, and $|Q_G(S,v)|$ can be computed in $\poly(n)$ time using, for example, depth-first search.

Then define
\[\TWR_G(L,S) = \min_{\substack{\pi \in \Pi(V) \\ \text{$L$ is a prefix of $\pi$}}} \max_{v \in S} |Q_G(L \cup \pi_{<v},v)|.\]
Also define
\[\TW_G(S) = \min_{\pi \in \Pi(V)} \max_{v \in S} |Q_G(\pi_{<v},v)|.\]
These notations are connected by the relation
\[\TW_G(S) = \TWR_G(\varnothing,S).\]
Note that $\tw(G)$ is equal to $\min_{\pi \in \Pi(V)} \max_{v \in V} R_\pi(v) = \TW_G(V)$.

The following lemma gives a way to find optimal fixed bag tree decompositions using the algorithms for finding the optimal linear arrangements:
\begin{lemma}[Lemma 11 in \cite{BFKAKT12}] \label{thm:clique}
Let $C \subseteq V$ induce a clique in a graph $G = (V,E)$.
The treewidth of $G$ equals $\max(\TW_G(V - C), |C|-1)$.
\end{lemma}
Essentially, this lemma tells us that $C$ can be placed in the end of the optimal arrangement.

\begin{lemma} \label{thm:suffix}
Let $G=(V,E)$ be a graph, and $\chi \subseteq V$ a subset of its vertices.
Then 
\[\tw(G,\chi) = \max\left(\TW_G(V-\chi),|\chi|-1\right).\]
\end{lemma}

\begin{proof}
Completing a bag of a tree decomposition into a clique does not change the width of the tree decomposition.
The claim then follows from Lemma \ref{thm:clique}.
\end{proof}

In the final treewidth algorithms, we will also use the following fact:

\begin{lemma} \label{thm:tw-comp}
Let $G=(V,E)$ be a graph and $\chi \subseteq V$ a subset of its vertices.
Let $\mathcal C$ be the set of connected components of $G[V - \chi]$.
Then 
\[\tw(G,\chi) = \max_{C \in \mathcal{C}} \tw(G[C \cup \chi],\chi).\]
\end{lemma}

\begin{proof}
Let $(X,T)$ be a tree decomposition with the smallest width $w$ that contains $\chi$ as a bag.
For a connected component $C \in \mathcal C$, examine the tree decomposition $(X_C,T_C)$ obtained from $(X,T)$ by removing all vertices not in $\chi$ or $C$ from all bags.
Clearly, this is a tree decomposition of $G[C \cup \chi]$ with $\chi$ as a bag; as we only have possibly removed some vertices, its width is at most $w$.
Now, examine the tree decomposition obtained by taking all $(X_C,T_C)$ and making $\chi$ its common bag.
This is a valid tree decomposition, since no two vertices in distinct connected components of $\mathcal C$ are connected by an edge.
Its width is the maximal width of $(X_C,T_C)$, therefore at most $w$.
\end{proof}

\subsection{Divide \& Conquer}

The first algorithm is based on the following property:
\begin{lemma}[Lemma 7 in \cite{BFKAKT12}] \label{thm:partition}
Let $G=(V,E)$ be a graph, $S \subseteq V$, $|S| \geq 2$, $L \subseteq V$, $L \cap S = \varnothing$, $1 \leq k < |S|$.
Then
\[\TWR_G(L,S) = \min_{\substack{S'\subseteq S\\|S'|=k}} \max\left( \TWR_G(L,S'),\TWR_G(L \cup S', S-S')\right).\]
\end{lemma}

Note that $\TWR_G(L,\{v\}) = |Q_G(L,v)|$ can be calculated in polynomial time.
The value we wish to calculate is $\TWR_G(\varnothing,V-\chi)$.
Picking $k = |S|/2$ in Lemma \ref{thm:partition} and applying Lemma \ref{thm:suffix}, we obtain a $\poly(|V|)4^{|V|-|\chi|}$ deterministic algorithm with polynomial space:
\begin{theorem}[Theorem 8 in \cite{BFKAKT12}] \label{thm:dq}
Let $G=(V,E)$ be a graph on $n$ vertices and $\chi \subseteq V$ a subset of its vertices.
There is an algorithm that calculates $\tw(G,\chi)$ in $O^*(4^{n-|\chi|})$ time and polynomial space.
\end{theorem}

Immediately we can prove a quadratic quantum speedup using Grover's search:
\begin{theorem} \label{thm:quantum-dq}
Let $G=(V,E)$ be a graph on $n$ vertices and $\chi \subseteq V$ a subset of its vertices.
There is a bounded-error quantum algorithm that calculates $\tw(G,\chi)$ in $O^*(2^{n-|\chi|})$ time and polynomial space.
\end{theorem}

\begin{proof}
We can apply quantum minimum finding to check sets $S'$ in Lemma \ref{thm:partition}, in order to obtain a quadratic speedup over Theorem \ref{thm:dq}.
To avoid the accumulation of error in the recursion, we can use the implementation of Grover's search with bounded-error inputs \cite{HMDw03}.
\end{proof}

\subsection{Dynamic programming}

The second algorithm is based on the following recurrence:
\begin{lemma}[Lemma 5 in \cite{BFKAKT12}] \label{thm:recur}
Let $G=(V,E)$ be a graph and $S \subseteq V$, $S \neq \varnothing$.
Then
\[ \TW_G(S) = \min_{v \in S} \max\left( \TW_G(S - \{v\}),|Q_G(S - \{v\},v)| \right).\]
\end{lemma}
Note that in fact Lemma \ref{thm:recur} is a special case of Lemma \ref{thm:partition} with $L=\varnothing$ and $k=|S|-1$.
This lemma together with Lemma \ref{thm:suffix} and the dynamic programming technique by Bellman, Held and Karp \cite{Bel62,HK62} gives the following algorithm:
\begin{theorem}[Theorem 6 in \cite{BFKAKT12}] \label{thm:dp}
Let $G=(V,E)$ be a graph on $n$ vertices and $\chi \subseteq V$ a subset of its vertices.
There is an algorithm that calculates $\tw(G,\chi)$ in $O^*(2^{n-|\chi|})$ time and space.
\end{theorem}

This algorithm calculates the values of $\TW_G(S)$ for all sets $S$ in order of increasing size of the sets, and also stores them all in memory.
Such dynamic programming can be sped up quantumly: Ambainis et al.~\cite{ABIKPV19} have shown a $O(1.817^n)$ time and space quantum algorithm with QRAM for such problems.
Therefore, this gives the following quantum algorithm:
\begin{theorem} \label{thm:quantum-dp}
Let $G=(V,E)$ be a graph on $n$ vertices and $\chi \subseteq V$ a subset of its vertices.
Assuming the QRAM data structure, there is a bounded-error quantum algorithm that calculates $\tw(G,\chi)$ in $O^*(1.81691^{n-|\chi|})$ time and space.
\end{theorem}

Note that this algorithm can be used to calculate $\TWR_G(L,S)$.
Firstly, $\TWR_G(L,\varnothing) = 0$ and
\begin{align*}
    \TWR_G(L,S)
    &= \min_{v \in S} \max\left( \TWR_G(L,S - \{v\}), \TWR_G(L \cup (S - \{v\}), \{v\}) \right)\\
    &= \min_{v \in S} \max\left( \TWR_G(L,S - \{v\}), |Q_G(L \cup (S - \{v\}), v)| \right)
\end{align*}
by Lemma \ref{thm:partition}.
As already mentioned earlier, the value $|Q_G(L \cup (S - \{v\}), v)|$ can be calculated in polynomial time.
Hence this recurrence is of the same form as Lemma \ref{thm:recur}.

\begin{theorem} \label{thm:quantum-twr}
Let $G=(V,E)$ be a graph on $n$ vertices and $L, S \subseteq V$ be disjoint subsets of vertices.
Assuming the QRAM data structure, there is a bounded-error quantum algorithm that calculates $\TWR_G(L,S)$ in $O^*(1.81691^{|S|})$ time and space.
\end{theorem}

\section{Fomin's and Villanger's algorithm} \label{sec:fva}

In this section, we first describe the polynomial space treewidth algorithm of \cite{FV12}.
Afterwards, we analyze the time complexity for the classical algorithm and then for the same algorithm sped up by the quantum tools presented above.

The algorithm relies on the following, shown implicitly in the proof of Theorem 7.3.~of \cite{FV12}.

\begin{lemma} 
Let $G = (V,E)$ be a graph, and $\beta \in \left[0,\frac 1 2\right]$.
There exists an optimal tree decomposition $(X,T)$ of $G$ so that at least one of the following holds:
\begin{enumerate}[(a)] \label{thm:dich}
    \item there exists a bag $\Omega \in X$ such that $\Omega$ is a potential maximum clique and there exists a connected component $C$ of $G[V - \Omega]$ such that $|C| \leq \beta n$;
    \item there exists a bag $S \in X$ such that $S$ is a minimal separator and there exist two disjoint connected components $C_1$, $C_2$ of $G[V - S]$ such that $N(C_1) = N(C_2) = S$ and $|C_2| \geq |C_1| \geq \beta n$.
\end{enumerate}
\end{lemma}

The idea of the algorithm then is to try out all possible potential maximum cliques and minimal separators that conform to the conditions of this lemma, and for each of these sets, to find an optimal tree decomposition of $G$ given that the examined set is a bag of the decomposition using the algorithm from Theorem \ref{thm:dq}.
The treewidth of $G$ then is the minimum width of all examined decompositions.

The potential maximum clique generation is based on the following lemma.
\begin{lemma}[Lemma 7.1.~in \cite{FV12}\footnote{The original lemma gives an upper bound if the size $\Omega$ is not fixed, but our statement follows from their proof. We need to fix $|\Omega|$ because in the quantum algorithms, Grover's search will be called for fixed $|\Omega|$ and $|C|$.}] \label{thm:pmc}
Let $G = (V,E)$ be a graph.
The number of maximum potential cliques $\Omega$ of size $p$ such that there exists a connected component $C$ of $G[V - \Omega]$ of size $c$ is at most $n \binom{n-c}{p-1}$.
The set of all these cliques can also be generated in time $O^*(\binom{n-c}{p-1})$.
\end{lemma}

For the minimal separators, suppose that the size of $S$ is fixed, denote it by $s$.
Note that since $C_1$ in Lemma \ref{thm:dich} is a connected component such that $N(C_1) = S$, then instead of generating minimal separators, we can generate the sets of vertices $C$ with neighborhood size equal to $s$.
The set of sets generated in this way contains all of the minimal separators of size $s$ that we are interested in, and for those sets that are not, the fixed-bag treewidth algorithm will still find some tree decomposition of the graph, albeit not an optimal one.
The generation is done using Lemma \ref{thm:comblem}: for a fixed size $c$ of $C$, the number of such $C$ with exactly $s$ neighbors is at most $n\binom{c+s}{c}$ (the factor of $n$ comes from trying each of $n$ vertices as the fixed vertex $v \in B$).
The algorithm generating all such $C$ requires time $O^*\left(\binom{c+s}{c}\right)$.
For a set $C$, we then find an optimal tree decomposition of $G$ containing $N(C)$ as a fixed bag using the algorithm from Theorem \ref{thm:dq}.
In this way we work through all $c$ from $\beta n$ to $n-s-|C_2| \leq (1-\beta)n-s$.

\begin{algorithm} 
\caption{The polynomial space algorithm for treewidth.}
\label{alg:classic-dq}
\begin{enumerate}
    \item For $c$ from $0$ to $\beta n$ and $p$ from $1$ to $n-c$ generate the set of potential maximal cliques $\Omega$ of size $p$ with a connected component of size $c$ using Lemma \ref{thm:pmc}. \label{itm:stage1}
    For each of these $\Omega$, find $\tw(G,\Omega)$ using Theorem \ref{thm:dq}.
    \item For $s$ from $1$ to $(1-2\beta)n$ and for $c$ from $\beta n$ to $(1-\beta)n-s$ generate the set of subsets $C$ such that $|C| = c$ and $N(C) = s$ using Lemma \ref{thm:comblem}. \label{itm:stage2}
    Let $S = N(C)$; then $\tw(G,S)$ is equal to the maximum of $\tw(G[S \cup C],S)$ and $\tw(G[V - C],S)$ by Lemma \ref{thm:tw-comp}.
    Use the algorithm from Theorem \ref{thm:dq} to compute these values.
    \item Output the minimum width of all examined tree decompositions.
\end{enumerate}
\end{algorithm}

\subsection{Classical complexity}

\begin{theorem}[Theorem 7.3.~in \cite{FV12}] \label{thm:tw-poly}
Algorithm \ref{alg:classic-dq} computes the treewidth of a graph with $n$ vertices in $O^*(2.61508^n)$ time and polynomial space.
\end{theorem}

\begin{proof}
The algorithms from Lemma \ref{thm:comblem} and Theorem \ref{thm:dq} both require polynomial space, hence it holds also for Algorithm \ref{alg:classic-dq}.

Now we analyze the time complexity; Stage \ref{itm:stage1} of the algorithm requires time
\begin{equation} \label{eq:cmpl1}
O^*\left(\sum_{c=0}^{\beta n} \sum_{p=1}^{n-c} \binom{n-c}{p-1} 4^{n-p}\right)
= O^*\left(\sum_{c=0}^{\beta n} 2^{n-c} 4^c\right)
= O^*\left(\max_{c=0}^{\beta n} 2^{n+c}\right)
= O^*\left(2^{(1+\beta)n}\right).
\end{equation}
Stage \ref{itm:stage2} of the algorithm requires time
\[O^*\left(\sum_{s=1}^{(1-2\beta)n} \sum_{c=\beta n}^{(1-\beta)n-s} \binom{c+s}{c} \max\left(4^{c},4^{n-c-s}\right)\right).\]
Note that we can assume that $C=C_1$ and $V - C - S$ contains $C_2$ (we can check this in polynomial time by finding the connected components of $G[V - S]$); since $|C_2| \geq |C_1|$, we can assume that $n-c-s \geq c$.
Hence the complexity becomes
\[O^*\left(\sum_{s=1}^{(1-2\beta)n} \sum_{c=\beta n}^{(1-\beta)n-s} \binom{c+s}{c} 4^{n-c-s}\right)
= O^*\left(\max_{s=1}^{(1-2\beta)n} \max_{c=\beta n}^{(1-\beta)n-s} \binom{c+s}{c} 4^{n-c-s}\right).\]
Now denote $d=n-c-s$, then $c+s=n-d$ and we can rewrite the complexity as
\[O^*\left(\max_{d=\beta n}^{(1-\beta)n} \max_{c=\beta n}^{n-d} \binom{n-d}{c} 4^d\right).\]
For any $d$, the maximum of $\binom{n-d}{c}$ over $c \geq \beta n$ can be one of two cases: if $\beta n \leq \frac{n-d}{2}$, it is equal to $\Theta^*(2^{n-d})$; otherwise it is equal to $\binom{n-d}{\beta n}$.
In the first case, for the interval $c \in \left[\beta n, \frac{n-d}{2}\right]$, the function being maximized becomes $2^{n-d}4^d=2^{n+d}$.
Since this function is increasing in $d$, its maximum is covered by the second case with the smallest $c$ such that $c = \frac{n-d}{2}$ (in case $\beta n \leq \frac{n-d}{2}$).
Therefore, the complexity of Stage  \ref{itm:stage2} of the algorithm becomes
\begin{equation} \label{eq:cmpl2}
O^*\left(\max_{d=\beta n}^{(1-\beta)n} \binom{n-d}{\beta n} 4^d\right).
\end{equation}
Now we are searching for the optimal $\beta \in \left[0,\frac{1}{2}\right]$ that balances the complexities (\ref{eq:cmpl1}) and (\ref{eq:cmpl2}).
We solve it numerically and obtain $\beta \approx 0.38685$, giving complexity $O^*(2.61508^n)$.
\end{proof}

\subsubsection{A time-space tradeoff}

One might ask whether replacing the $O^*(4^n)$ divide \& conquer algorithm from Theorem \ref{thm:dq} with the $O^*(2^n)$ dynamic programming algorithm from Theorem \ref{thm:dp} in Algorithm \ref{alg:classic-dq} can give any interesting complexity.
Indeed, we can show the following previously unexamined classical time-space tradeoff.

\begin{theorem} \label{thm:tradeoff}
The treewidth of a graph with $n$ vertices can be computed in $O^*(2^n)$ time and $O^*\left(\sqrt{2^n}\right)$ space.
\end{theorem}

\begin{proof}
First, we look at the time complexity.
The time complexity of Stage \ref{itm:stage1} now is equal to
\[O^*\left(\sum_{c=0}^{\beta n} 2^{n-c} 2^c\right) = O^*(2^n).\]
The time complexity of Stage \ref{itm:stage2} is equal to
\[O^*\left(\max_{d=\beta n}^{(1-\beta)n}  \binom{n-d}{\beta n} 2^d \right) = O^*\left(2^{n-d}2^d\right) = O^*(2^n).\]
The space complexity of Stage \ref{itm:stage1} is equal to
\[O^*\left(\max_{c=0}^{\beta n} 2^c \right) = O^*\left(2^{\beta n}\right).\]
The space complexity of Stage \ref{itm:stage2} is equal to
\[O^*\left(\max_{d=\beta n}^{(1-\beta)n} 2^d \right) = O^*\left(2^{(1-\beta)n}\right).\]
Therefore, the time complexity of this algorithm is $O^*(2^n)$ and, taking $\beta = \frac{1}{2}$, the space complexity is equal to $O^*\left(\sqrt{2^n}\right)$.
\end{proof}

We can compare this to the existing treewidth algorithms.
The most time-efficient treewidth algorithm runs in time and space $O^*(1.755^n)$ \cite{FV12}, which is more than $O^*\left(\sqrt{2^n}\right)$.
The polynomial space $O^*(2.616^n)$ algorithm, of course, is slower than $O^*(2^n)$.
The time-space tradeoffs for permutation problems from \cite{KP10} give $TS \gtrsim 3.93$, where $T \geq 2$ and $\sqrt{2} \leq S \leq 2$ are the time and space complexities (bases of the exponent $n$) of the algorithm.
In this case, $TS = 2^{\frac{3}{2}} \approx 2.83$, $T = 2$ and $S=\sqrt{2}$.
Therefore, Theorem \ref{thm:tradeoff} fully subsumes their tradeoff for treewidth.
We also note that we cannot ``tune'' our tradeoff directly for less time and more space, since the first stage with $c=0$ already requires $\Theta^*(2^n)$ time for any $\beta$.

\subsection{Quantum complexity}

Now we are ready to examine the quantum versions of the algorithm.
First, we consider the analogue of Algorithm \ref{alg:classic-dq} with its procedures sped up quadratically using Grover's search.

\begin{theorem} \label{thm:qtwdq}
There is a bounded-error quantum algorithm that finds the treewidth of a graph on $n$ vertices in $O(1.61713^n)$ time and polynomial space.
\end{theorem}

\begin{proof}
In Algorithm \ref{alg:classic-dq}, we replace the algorithms from Lemmas \ref{thm:comblem} and \ref{thm:pmc} with the quantum algorithm from Lemma \ref{thm:quantum-comb}; the algorithm from Theorem \ref{thm:dq} is replaced with the algorithm from Theorem \ref{thm:quantum-dq}.
Since all exponential subprocedures now are sped up quadratically, the time complexity becomes
\[O\left(\sqrt{2.61508^n}\right) = O(1.61713^n).\]
The space complexity is still polynomial, since Grover's search additionally uses only polynomial space.
\end{proof}

Similarly, we can replace the algorithm from Theorem \ref{thm:dp} with the quantum dynamic programming algorithm from Theorem \ref{thm:quantum-dp}:

\begin{theorem} \label{thm:qtwdp}
Assuming the QRAM data structure, there is a bounded-error quantum algorithm that finds the treewidth of a graph on $n$ vertices in $O(1.55374^n)$ time and $O(1.45195^n)$ space.
\end{theorem}

\begin{proof}
The time complexity of the first stage is now equal to
\[O^*\left(\sum_{c=0}^{\beta n} \sqrt{2^{n-c}} \cdot 1.81691^c\right) = O^*\left(\sqrt{2^n} \cdot 1.28475^{\beta n}\right).\]
For the second stage, the time is given by
\[O^*\left(\max_{d=\beta n}^{(1-\beta)n}  \sqrt{\binom{n-d}{\beta n}} \cdot 1.81691^d  \right).\]
We can numerically find that $\beta \approx 0.3755$ balances these complexities, which then are both $O(1.55374^n)$.
The space complexity is
\[O^*\left(1.81691^{\max(\beta n,(1-\beta)n)}\right) = O^*\left(1.81691^{(1-\beta)n}\right) = O(1.45195^n).\qedhere\]
\end{proof}

\section{Improved quantum algorithm} \label{sec:main}

We can see that in Theorem \ref{thm:qtwdp} we still have some room for improvement by trading space for time.
This can be done using an additional technique.
The main idea is to make a global precalculation for $\TW_G(S)$ for all subsets $S \subseteq V$ of size at most $\alpha n$, for some constant parameter $\alpha$.
Then, as we will see later, these values can be used in all calls of the quantum dynamic programming because of the properties of treewidth.
For many such calls, this reduces the $O^*(1.817^d)$ running time to something smaller, which in turn reduces the overall time complexity.

\subsection{Asymmetric quantum dynamic programming on the hypercube}

We describe our modification to the quantum dynamic programming algorithm by Ambainis et al.\, \cite{ABIKPV19}.
First, we prove the following lemma that allows us to reutilize the precalculated DP values on the original graph $G$ in the DP calculation in the subgraphs examined by our algorithms.

\begin{lemma} \label{thm:global}
Let $G=(V,E)$ be a graph, and $\chi \subseteq V$ a subset of its vertices.
Suppose that $C$ is a union of some connected components of $G[V - \chi]$.
Then for any $S \subseteq C$, we have $\TW_{G[C \cup \chi]}(S) = \TW_G(S)$.
\end{lemma}

\begin{proof}
Examine the permutations $\pi$ achieving $\TW_{G[C \cup \chi]}(S) = \min_{\pi \in \Pi(C \cup \chi)} \max_{v \in S} |Q_G(\pi_{<v},v)|$.
As a direct consequence of Lemma \ref{thm:recur}, there exists such a permutation $\pi$ with the property that $S$ is its prefix.
Now let $\pi' \in \Pi(V)$ be a permutation obtained by adding the vertices of $V - C - \chi$ at the end of $\pi$ in any order.
Examine any vertex $u \in V - C - \chi$ and any $v \in S$.
Since $u$ and $v$ are located in different connected components of $G[C - \chi]$, any path from $u$ to $v$ in $G$ passes through some vertex of $\chi$.
However, $\pi_{<v} \cap \chi = \varnothing$, as $\pi_{<v} \subset S$.
Then we can conclude that $Q_{G[C \cup \chi]}(\pi_{<v},v) = Q_G(\pi'_{<v},v)$, as $u$ cannot contribute to $Q$.
Therefore, $\TW_G(S) \leq \TW_{G[C \cup \chi]}(S)$.
On the other hand, $\TW_G(S) \geq \TW_{G[C \cup \chi]}(S)$, as additional vertices cannot decrease $\TW$.
\end{proof}

Now we are ready to describe our quantum dynamic programming procedure.
Suppose that all values of $\TW_G(S)$ for sets with $|S| \leq \alpha n$ have been precalculated beforehand and stored in QRAM, where $\alpha \in \left[0,\frac{1}{2}\right]$ is some fixed parameter.
Suppose that we have fixed a subset $\chi \subseteq V$, and our task is to calculate $\tw(G[C \cup \chi],\chi)$ for a union $C$ of some connected components of $G[V - \chi]$.
By Lemma \ref{thm:suffix}, it is equal to $\max(\TW_{G[C \cup \chi]}(C),|\chi|-1)$.
Since $|\chi|$ is known, our goal is to compute $\TW_{G[C \cup \chi]}(C)$.

Let $n=|V|$ and $n'=|C|$.
If $n' \leq \alpha n$, then $\TW_{G[C \cup \chi]}(C) = \TW_G(C)$ by Lemma \ref{thm:global} and is known from the precalculated values.
Hence, assume that $n' > \alpha n$.
Pick some natural $k$, we will call this \emph{the number of layers}.
Let $\lambda_1= \frac{\alpha n}{n'}$, and pick constants $\lambda_2 < \ldots < \lambda_k < \mu < \rho_k < \ldots < \rho_1 < 1$, with $\lambda_1 < \lambda_2$.
Then define collections
\begin{align*}
    \mathcal L_i &= \{S \subseteq C \mid |S| = \lambda_i n'\},\\ 
    \mathcal M &= \{S \subseteq C \mid |S| = \mu n'\},\\ 
    \mathcal R_i &= \{S \subseteq C \mid |S| = \rho_i n'\}.
\end{align*}
We call these collections \emph{layers}: we can represent subsets $S \subseteq C$ as vertices on the hypercube of dimension $n'$; then these layers are defined as the subsets of vertices with some fixed Hamming weight, see Figure \ref{fig:hypercube}.
For all sets $S$ corresponding to the vertices in the crosshatched area (such that $|S| \leq \alpha n$), the value of $\TW_{G[C \cup \chi]}(S) = \TW_G(S)$ is known from the assumed precalculation.

\begin{figure}[H]
    \centering
    \includegraphics[scale=0.7]{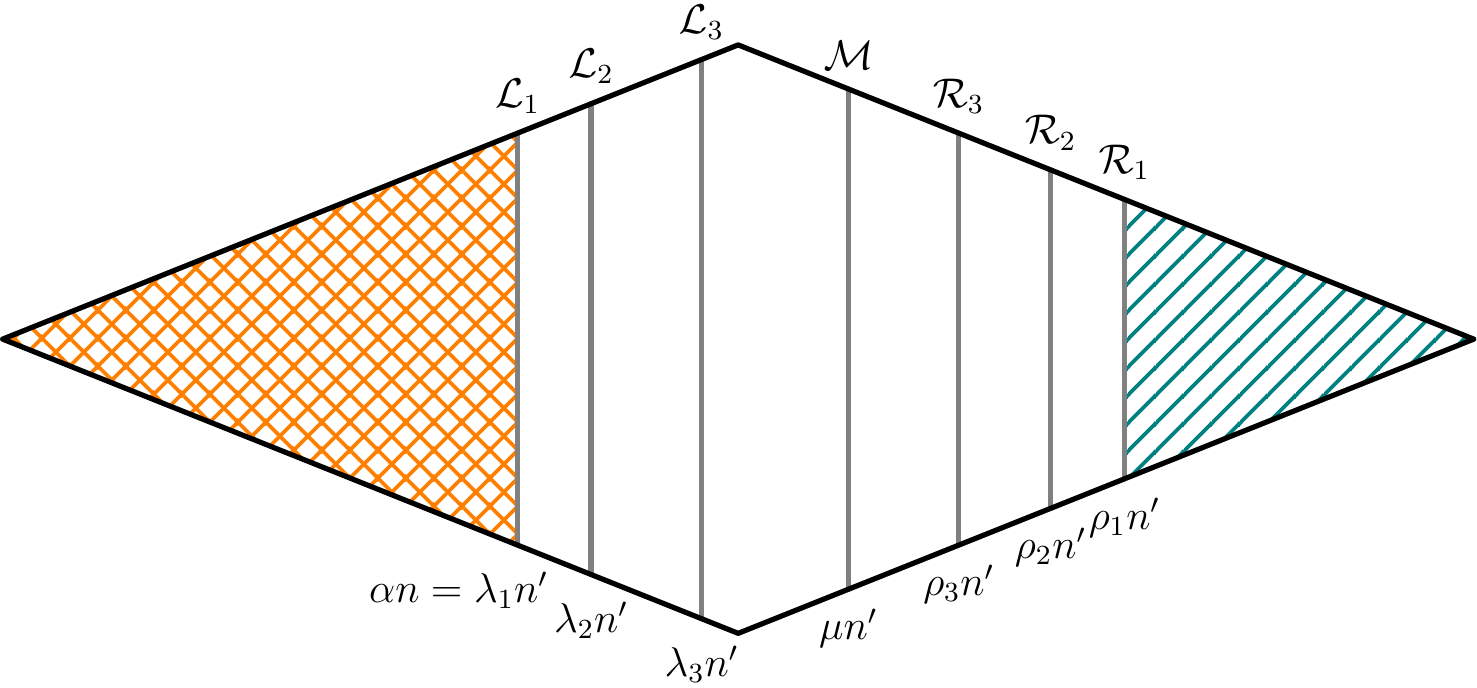}
    \caption{A schematic representation of layers in the Boolean hypercube with $k=3$.}
    \label{fig:hypercube}
\end{figure}

Now we will describe the quantum procedure.
Denote $G'=G[C \cup \chi]$.
Also denote $\TW'_{G'}(S) = \TWR_{G'}(S,C - S)$ and note that $\TW_{G'}(S) = \TWR_{G'}(\varnothing,S)$.
Informally, calculating $\TW_{G'}(S)$ means finding the best ordering for the vertices $S$ as a prefix of the permutation, and $\TW'_{G'}(S)$ means finding the best ordering for the vertices $C - S$, where $C - S$ is in the middle of permutation, followed by some ordering of $\chi$.

Algorithm \ref{alg:asymdp} is exactly the algorithm of \cite{ABIKPV19}, with the exception that the precalculation is performed only for suffixes (and the precalculation for prefixes comes ``for free'').
The idea is to find the optimal path between the vertices $s$ and $t$ with the smallest and highest Hamming weight in the hypercube.
First, we use Grover's search over the vertex $v_{k+1}$ in the middle layer $\mathcal M$.
Then we search independently for the best path from $s$ to $v_{k+1}$ and from $v_{k+1}$ to $t$; the optimal path from $s$ to $t$ is their concatenation.
To find the best path from $s$ to $v_{k+1}$, we use Grover's search over the vertex $v_k$ on the layer $\mathcal L_k$ such that there exists a path from $v_k$ to $v_{k+1}$.
Then we find the best path from $v_k$ to $v_{k+1}$ by recursively using the $O^*(1.817^{n'})$ algorithm (where $n'$ is the dimension of the hypercube with $v_k$ and $v_{k+1}$ being the smallest and largest weight vertices, respectively).
We combine it with the best path from $s$ to $v_k$, which we find in the similar way (fixing $v_{k-1}$, $\ldots$, $v_1$).
The value of the optimal path from $s$ to $v_1$ is known from the global precalculation we assumed took place before the algorithm.
The optimal path from $v_{k+1}$ to $t$ is found analogously; only to know the value of the best path from vertices in $\mathcal R_1$ to $t$, we have to precalculate these values ``from the back'' using Bellman \& Held-Karp dynamic programming in the beginning of the algorithm.
The formal description of this algorithm for treewidth is given in Algorithm \ref{alg:asymdp}.

\bigskip
\begin{breakablealgorithm}
\caption{Asymmetric quantum dynamic programming algorithm.}\label{alg:asymdp}
\begin{enumerate}
    \item \label{itm:prec} For all $S \in \mathcal R_1$, calculate and store in QRAM the values $\TW'_{G'}(S)$ using the recurrence
    \[\TW'_{G'}(S) = \min_{v \in C - S}  \max\left(\TW'_{G'}(S \cup \{v\}), |Q_{G'}(S,v)|\right) .\]
    This follows from Lemma \ref{thm:partition} with $k=1$.
    \item \label{itm:qmf} Use quantum minimum finding over sets $S \in \mathcal M$ to find the answer,
    \[\TW_{G'}(C) = \min_{S \in \mathcal M} \max\left(\TW_{G'}(S),\TW'_{G'}(S)\right).\]
    This also follows from Lemma \ref{thm:partition} with $k=\mu n'$.
    
    \begin{itemize}
        \item To find $\TW_{G'}(S)$, we use the recursive procedure $\textsc{BestPrefix}_i(G',S)$.
        Its value is equal to $\TW_{G'}(S)$, and it requires $S \in \mathcal L_i$ (if $i=k+1$, then $S \in \mathcal M$).
        The needed value is then given by $\textsc{BestPrefix}_{k+1}(G',S)$.
        The description of $\textsc{BestPrefix}_i(G',S)$:
        \begin{itemize}
            \item If $i=1$, return $\TW_{G'}(S)=\TW_G(S)$ that is stored in QRAM.
            \item If $1 < i \leq k+1$, then use quantum minimum finding over the sets $T \in \mathcal L_{i-1}$ to find
            \[\TW_{G'}(S) = \min_{\substack{T \in \mathcal L_{i-1} \\ T \subset S}} \max\left(\textsc{BestPrefix}_{i-1}(G',T),\TWR_{G'}(T,S - T)\right).\]
            Again, this recurrence follows from Lemma \ref{thm:partition} with $k = \lambda_{i-1} n'$.
            The value of $\TWR_{G'}(T,S - T)$ is calculated by the quantum dynamic programming from Theorem \ref{thm:quantum-twr} and requires $O^*(1.817^{|S| - |T|})$ time and QRAM space.
        \end{itemize}
        \item To find $\TW'_{G'}(S)$, we similarly use the recursive procedure $\textsc{BestSuffix}_i(G',S)$.
        Its value is equal to $\TW'_{G'}(S)$, and it requires $S \in \mathcal R_i$ (if $i=k+1$, then $S \in \mathcal M$).
        The needed value is then given by $\textsc{BestSuffix}_{k+1}(G',S)$.
        The description of $\textsc{BestSuffix}_i(G',S)$:
        \begin{itemize}
            \item If $i=1$, return $\TW'_{G'}(S)$ stored in QRAM from the precalculation in Step \ref{itm:prec}.
            \item If $1 < i \leq k+1$, then use quantum minimum finding over the sets $T \in \mathcal R_{i-1}$ to find
            \[\TW'_{G'}(S) = \min_{\substack{T \in \mathcal R_{i-1} \\ S \subset T}} \max\left(\TWR_{G'}(S,T - S),\textsc{BestSuffix}_{i-1}(G',T)\right).\]
            Again, this recurrence follows from Lemma \ref{thm:partition} with $k = \rho_{i-1}n' - \rho_in'$.
            The value of $\TWR_{G'}(S,T - S)$ is calculated by the quantum dynamic programming from Theorem \ref{thm:quantum-twr} and requires $O^*(1.817^{|T| - |S|})$ time and QRAM space.
        \end{itemize}
    \end{itemize}
\end{enumerate}
\end{breakablealgorithm}
\newpage

Finally, note that with $\lambda_1 \approx 0.28448$, the time complexity of Algorithm \ref{alg:asymdp} becomes $O^*(1.817^{n'})$, as this is the same parameter for the precalculation layer as in \cite{ABIKPV19}.
Thus if it happens that $\alpha n < 0.28448 n'$, the asymmetric version of the algorithm will have time complexity larger than $O^*(1.817^{n'})$, so in that case it is better to call the algorithm from Theorem \ref{thm:quantum-dp}.
Therefore, our procedure for calculating $\TW_{G[C \cup \chi]}(C)$ is as follows:

\begin{algorithm} 
\caption{Quantum algorithm calculating $\tw(G[C \cup \chi],\chi)$ assuming global precalculation.}
\label{alg:subg}
Assume that $\TW_G(S)$ are stored in QRAM for all $|S| \leq \alpha n$.
\begin{itemize}
    \item If $n' \leq \alpha n$, fetch $\TW_{G[C \cup \chi]}(C) = \TW_G(C)$ from the global precalculation.
    \item Else if $\alpha n \leq 0.28448 n'$, find $\TW_{G[C \cup \chi]}(C)$ using the $O(1.817^{n'})$ algorithm of Theorem \ref{thm:quantum-dp}.
    \item Else calculate $\TW_{G[C \cup \chi]}(C)$ using Algorithm \ref{alg:asymdp}.
\end{itemize}
Return $\tw(G[C \cup \chi],\chi) = \max\left(\TW_{G[C \cup \chi]}(C), |\chi|-1\right)$.
\end{algorithm}

\subsection{Complexity of the quantum dynamic programming}

We will estimate the time complexity of Algorithm \ref{alg:asymdp}. The space complexity will not be necessary, because for the final treewidth algorithm it will be dominated by the global precalculation, as we will see later.
\begin{itemize}
    \item The time of the precalculation Step \ref{itm:prec} is dominated by the size of the layer $\mathcal R_1$.
    It is equal to $O^*\left(|\mathcal R_1|\right) = O^*(\binom{n'}{\rho_1 n'})$, which by Lemma \ref{thm:entropy} is
    \[O^*\left(2^{\be(\rho_1) n'}\right).\]
    \item Let the time of a call of $\textsc{BestPrefix}_i(G',S)$ be $T_i$, it can be calculated as follows.
    If $i=1$,
    \[T_1 = O^*(1),\]
    as all we need to do is to fetch the corresponding value $\TW_G(S)$ from QRAM.
    If $i>1$, then quantum minimum finding examines all $T \in \mathcal L_{i-1}$ such that $T \subset S$.
    The number of such $T$ is $\binom{|S|}{|T|} = \binom{\lambda_{i}n'}{\lambda_{i-1}n'}$ (for generality, denote $\lambda_{k+1}=\mu$).
    Again, by Lemma \ref{thm:entropy}, this is at most $2^{\be(\lambda_{i-1}/\lambda_{i}) \cdot \lambda_i n'}$.
    The call to $\textsc{BestPrefix}_{i-1}(G',T)$ requires time $T_{i-1}$ and  calculating $\TWR_{G'}(T,S - T)$ with the algorithm from Theorem \ref{thm:quantum-twr} requires time $O^*(1.817^{|S| - |T|}) = O^*(1.817^{(\lambda_i-\lambda_{i-1})n'})$.
    Putting these estimates together, we get that for $i>1$,
    \[T_i = O^*\left(\sqrt{2^{\be\left(\frac{\lambda_{i-1}}{\lambda_{i}}\right) \cdot \lambda_i n'}} \cdot \max\left( T_{i-1}, 1.817^{(\lambda_i-\lambda_{i-1})n'}\right)\right).\]
    \item The time $T'_i$ for $\textsc{BestSuffix}_i(G',S)$ is calculated analogously.
    We can check the precalculated values from Step \ref{itm:prec} in
    \[T'_1 = O^*(1)\]
    and (taking $\rho_{k+1} = \mu$) for $i > 1$,
    \[T'_i = O^*\left(\sqrt{2^{\be\left(\frac{1-\rho_{i-1}}{1-\rho_{i}}\right) \cdot (1-\rho_i) n'}} \cdot \max\left( T'_{i-1}, 1.817^{(\rho_{i-1}-\rho_i)n'}\right)\right).\]
    \item Lastly, the number of sets examined in the first quantum minimum finding in Step \ref{itm:qmf} is equal to the size of $\mathcal M$, which is $\binom{n'}{\mu n'} = 2^{\be(\mu)n'}$ by Lemma \ref{thm:entropy}.
    Therefore, Step \ref{itm:qmf} requires time
    \[O^*\left( \sqrt{2^{\be(\mu)n'}} \cdot \max\left( T_{k+1}, T'_{k+1}\right)\right).\]
\end{itemize}

For any of the complexities $\mathcal T$ examined here, let's look at $\log_2(\mathcal T)/n'$; since we are interested in the exponential complexity, we need to investigate only the constant $c$ in $O^*(2^{cn})$.
Also note that $\log_2(1.817) \approx 0.862$.
This results in the following optimization program
\begin{align*}
    \text{minimize}&\hspace{1cm}T(\lambda_1)=\max\left( \be(\rho_1), \frac{\be(\mu)}{2} + \max\left(t_{k+1},t'_{k+1}\right) \right) \\
    \text{subject to}
    &\hspace{1cm} \lambda_1 < \ldots < \lambda_k < \lambda_{k+1} = \mu = \rho_{k+1} < \rho_k < \ldots  < \rho_1 < 1\\
    &\hspace{1cm} t_i = \frac{1}{2}\cdot \be\left(\frac{\lambda_{i-1}}{\lambda_{i}}\right) \cdot \lambda_i + \max\left( t_{i-1}, 0.862(\lambda_i - \lambda_{i-1})\right) \quad &\text{for $i \in [2,k+1]$}\\
    &\hspace{1cm} t_1 = 0\\
    &\hspace{1cm} t'_i = \frac{1}{2}\cdot\be\left(\frac{1-\rho_{i-1}}{1-\rho_{i}}\right) \cdot (1-\rho_i) + \max\left( t'_{i-1}, 0.862(\rho_{i-1} - \rho_i)\right) \quad &\text{for $i \in [2,k+1]$}\\
    &\hspace{1cm} t'_1 = 0
\end{align*}

We can solve this program numerically and find the time complexity, depending on the value of $\lambda_1$.
Note that for $\lambda_1 \leq 0.28448$ the $O(1.817^{n'})$ symmetric quantum dynamic programming is more efficient, so we don't have to calculate the complexity in that case.
Figure \ref{fig:graph} shows the time complexity $T(\lambda_1)$ for $k=0,1,2,3$.
We can see that the advantage of adding additional layers quickly becomes negligible.

\begin{figure}[H]
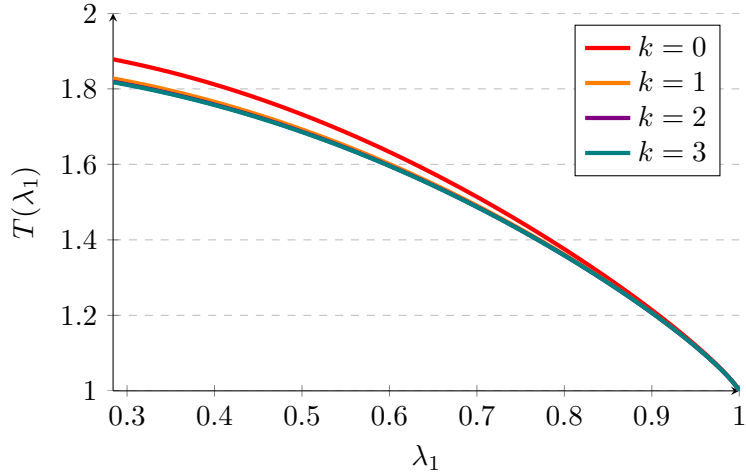

\begin{center}

\end{center}
\caption{Running time of the asymmetric quantum dynamic programming algorithm.} \label{fig:graph}
\end{figure}

\subsection{Final quantum algorithm}

Now we can give the improved quantum dynamic programming algorithm for treewidth.
It requires two constant parameters: $\alpha, \beta \in \left[0,\frac{1}{2}\right]$.
The value $\alpha n$ gives the limit for the global precalculation, and $\beta n$ is the cutoff point for the two stages as in Algorithm \ref{alg:classic-dq}.

\begin{algorithm}
\caption{Improved quantum algorithm for treewidth.}
\begin{enumerate}
    \item \label{itm:iqt1} Calculate $\TW_G(S)$ for all subsets $S$ such that $|S| \leq \alpha n$ and store them in QRAM.
    \item \label{itm:iqt2} For $c$ from $0$ to $\beta n$ and $p$ from $1$ to $n-c$ examine the set of potential maximal cliques $\Omega$ of size $p$ with a connected component of size $c$.
    Apply Lemma \ref{thm:quantum-comb} to Lemma \ref{thm:pmc} to find the minimum of $\tw(G,\Omega)$ in $O^*\left(\sqrt{\binom{n-c}{p-1}}\right)$ iterations.
    Calculate the value of $\tw(G,\Omega)$ using Algorithm \ref{alg:subg}.
    \item \label{itm:iqt3} For $s$ from $1$ to $(1-2\beta)n$ and from $c$ from $\beta n$ to $(1-\beta)n-s$ examine the set of subsets $C$ such that $|C|=c$ and $N(C)=s$.
    Let $S = N(C)$; then $\tw(G,S)$ is equal to the maximum of $\tw(G[S \cup C],S)$ and $\tw(G[V - C],S)$ by Lemma \ref{thm:tw-comp}.
    Find the minimum of $\tw(G,S)$ using Lemma \ref{thm:quantum-comb} in $O^*\left(\sqrt{\binom{c+s}{s}}\right)$ iterations.
    Calculate the values of $\tw(G[S \cup C],S)$ and $\tw(G[V - C],S)$ using Algorithm \ref{alg:subg}.
    \item \label{itm:iqt4} Return the minimum width of all examined tree decompositions.
\end{enumerate}
\end{algorithm}

We can now calculate the complexity similarly as in Theorem \ref{thm:qtwdp}.

\begin{theorem} \label{thm:main}
Assuming the QRAM data structure, there is a bounded-error quantum algorithm that finds the exact treewidth of a graph on $n$ vertices in $O(1.53793^n)$ time and space.
\end{theorem}

\begin{proof}
First, we choose $\alpha$ such that it balances the time complexity of the global precalculation (Step \ref{itm:iqt1}) and the rest of the algorithm (Steps \ref{itm:iqt2}--\ref{itm:iqt4}).
The space complexity of this step asymptotically is equal to its time complexity.
Therefore, the space complexity of this algorithm is equal to
$O^*\left(\binom{n}{\alpha n}\right).$

Denote the time complexity of Algorithm \ref{alg:subg} with $|C|=n'$ and some chosen $\lambda_1$ by $O^*(\mathcal T(\lambda_1)^{n'})$.
For fixed $\alpha$ and $n'$, $\lambda_1$ is calculated as $\alpha n / n'$.
Then
\[
    \mathcal T(\lambda_1) = 
    \begin{cases}
    1,&\text{if $\lambda_1 \geq 1$,}\\
    1.81691,&\text{if $\lambda_1 \leq 0.28448$,}\\
    T(\lambda_1),&\text{otherwise.}
    \end{cases}
\]
Similarly as we have obtained Equations (\ref{eq:cmpl1}, \ref{eq:cmpl2}) in the proof of Theorem \ref{thm:tw-poly}, we can also calculate the time complexity here.
The time complexity of Step \ref{itm:iqt2} now is equal to
\[O^*\left(\sum_{c=0}^{\beta n} \sqrt{2^{n-c}} \cdot \mathcal T\left( \frac{\alpha n}{c} \right)^c \right).\]
The running time of Step \ref{itm:iqt3} is given by
\[O^*\left(\max_{d=\beta n}^{(1-\beta)n} \sqrt{\binom{n-d}{\beta n}} \cdot \mathcal T\left( \frac{\alpha n}{d} \right)^d\right).\]
We can numerically find that $\alpha \approx 0.15447$ and $\beta \approx 0.38640$ balance these complexities, which then are all $O(1.53793^n)$.
In our numerical calculation, we have used $k=3$ for Algorithm \ref{alg:subg}.
\end{proof}

\section{Acknowledgements}

This work was supported by the project ``Quantum algorithms: from complexity theory to experiment'' funded under ERDF programme 1.1.1.5.

\printbibliography

@article{FLSPW18,
    author = {Fomin, Fedor V. and Lokshtanov, Daniel and Saurabh, Saket and Pilipczuk, Michał and Wrochna, Marcin},
    title = {Fully Polynomial-Time Parameterized Computations for Graphs and Matrices of Low Treewidth},
    year = {2018},
    publisher = {Association for Computing Machinery},
    address = {New York, NY, USA},
    volume = {14},
    number = {3},
    issn = {1549-6325},
    doi = {10.1145/3186898},
    journal = {ACM Trans. Algorithms},
    month = {6},
    articleno = {34},
    numpages = {45},
	archivePrefix = {arXiv},
	eprint    = {1511.01379},
	primaryClass = "cs.DS"
}

@article{FHL08,
    author = {Feige, Uriel and Hajiaghayi, Mohammad Taghi and Lee, James R.},
    title = {Improved Approximation Algorithms for Minimum Weight Vertex Separators},
    journal = {SIAM Journal on Computing},
    volume = {38},
    number = {2},
    pages = {629-657},
    year = {2008},
    doi = {10.1137/05064299X},
}

@misc{Kor21,
    author={Tuukka Korhonen},
    title={A Single-Exponential Time 2-Approximation Algorithm for Treewidth},
    year={2021},
	archivePrefix = {arXiv},
	eprint    = {2104.07463},
	primaryClass = "cs.DS"
}

@article{ACP87,
    author = {Arnborg, Stefan and Corneil, Derek G. and Proskurowski, Andrzej},
    title = {Complexity of Finding Embeddings in a $k$-Tree},
    journal = {SIAM Journal on Algebraic Discrete Methods},
    volume = {8},
    number = {2},
    pages = {277-284},
    year = {1987},
    doi = {10.1137/0608024},
}

@InProceedings{Bod05,
    author="Bodlaender, Hans L.",
    editor="Vojt{\'a}{\v{s}}, Peter
    and Bielikov{\'a}, M{\'a}ria
    and Charron-Bost, Bernadette
    and S{\'y}kora, Ondrej",
    title="Discovering Treewidth",
    booktitle="SOFSEM 2005: Theory and Practice of Computer Science",
    year="2005",
    publisher="Springer Berlin Heidelberg",
    %address="Berlin, Heidelberg",
    pages="1--16",
    isbn="978-3-540-30577-4",
    doi="10.1007/978-3-540-30577-4_1"
}

@article{AP89,
    author = {Arnborg, Stefan and Proskurowski, Andrzej},
    title = {Linear Time Algorithms for NP-Hard Problems Restricted to Partial $k$-Trees},
    year = {1989},
    issue_date = {April 1989},
    publisher = {Elsevier Science Publishers B. V.},
    address = {NLD},
    volume = {23},
    number = {1},
    issn = {0166-218X},
    doi = {10.1016/0166-218X(89)90031-0},
    journal = {Discrete Appl. Math.},
    month = {4},
    pages = {11–24},
    numpages = {14}
}

@InProceedings{MIKL20,
    author="Miyamoto, Masayuki
    and Iwamura, Masakazu
    and Kise, Koichi
    and Le Gall, Fran{\c{c}}ois",
    editor="Kim, Donghyun
    and Uma, R. N.
    and Cai, Zhipeng
    and Lee, Dong Hoon",
    title="Quantum Speedup for the Minimum Steiner Tree Problem",
    booktitle="Computing and Combinatorics",
    year="2020",
    publisher="Springer International Publishing",
    address="Cham",
    pages="234--245",
    isbn="978-3-030-58150-3",
    doi="10.1007/978-3-030-58150-3_19",
	archivePrefix = {arXiv},
	eprint    = {1904.03581},
	primaryClass = "quant-ph"
}

@InProceedings{Tan20,
    author =	{Seiichiro Tani},
    title =	{{Quantum Algorithm for Finding the Optimal Variable Ordering for Binary Decision Diagrams}},
    booktitle =	{17th Scandinavian Symposium and Workshops on Algorithm Theory (SWAT 2020)},
    pages =	{36:1--36:19},
    series =	{Leibniz International Proceedings in Informatics (LIPIcs)},
    ISBN =	{978-3-95977-150-4},
    ISSN =	{1868-8969},
    year =	{2020},
    volume =	{162},
    editor =	{Susanne Albers},
    publisher =	{Schloss Dagstuhl--Leibniz-Zentrum f{\"u}r Informatik},
    address =	{Dagstuhl, Germany},
    doi =		{10.4230/LIPIcs.SWAT.2020.36},
	archivePrefix = {arXiv},
	eprint    = {1909.12658},
	primaryClass = "quant-ph"
}

@inbook{KP10,
    author = {Mikko Koivisto and Pekka Parviainen},
    title = {A Space–Time Tradeoff for Permutation Problems},
    booktitle = {Proceedings of the 2010 Annual ACM-SIAM Symposium on Discrete Algorithms},
    year = {2010},
    pages = {484-492},
    publisher = {Society for Industrial and Applied Mathematics},
    numpages = {9},
    address = {USA},
    location = {Austin, Texas},
    series = {SODA '10},
    doi = {10.1137/1.9781611973075.41}
}

@article{BFKAKT12,
    author = {Bodlaender, Hans L. and Fomin, Fedor V. and Koster, Arie M. C. A. and Kratsch, Dieter and Thilikos, Dimitrios M.},
    title = {On exact algorithms for treewidth},
    year = {2012},
    issue_date = {December 2012},
    publisher = {Association for Computing Machinery},
    address = {New York, NY, USA},
    volume = {9},
    number = {1},
    issn = {1549-6325},
    doi = {10.1145/2390176.2390188},
    journal = {ACM Trans. Algorithms},
    articleno = {12},
    numpages = {23}
}

@InProceedings{SM20,
    author="Shimizu, Kazuya
    and Mori, Ryuhei",
    editor="Kohayakawa, Yoshiharu
    and Miyazawa, Flávio Keidi",
    title="Exponential-Time Quantum Algorithms for Graph Coloring Problems",
    booktitle="LATIN 2020: Theoretical Informatics",
    year="2020",
    publisher="Springer International Publishing",
    address="Cham",
    pages="387--398",
    doi="10.1007/978-3-030-61792-9_31",
	archivePrefix = {arXiv},
	eprint    = {1907.00529},
    primaryClass = {cs.DS}
}

@article{FV12,
    title={Treewidth computation and extremal combinatorics},
    author={Fedor~V. Fomin and Yngve Villanger},
    journal={Combinatorica},
    year={2012},
    volume={32},
    pages={289-308},
    doi={10.1007/s00493-012-2536-z},
	archivePrefix = {arXiv},
    eprint = {0803.1321},
    primaryClass = {cs.DS}
}

@inproceedings{HMDw03,
    author = {H\o{}yer, Peter and Mosca, Michele and de Wolf, Ronald},
    title = {Quantum Search on Bounded-Error Inputs},
    year = {2003},
    isbn = {3540404937},
    booktitle = {Automata, Languages and Programming},
    publisher = {Springer-Verlag},
    address = {Berlin, Heidelberg},
    pages = {291–299},
    numpages = {9},
    venue = {Eindhoven, The Netherlands},
    doi = {10.1007/3-540-45061-0_25},
    series = {ICALP'03},
	archivePrefix = {arXiv},
    eprint = {quant-ph/0304052}
}

@inproceedings{Grover96,
    author = {Grover, Lov K.},
    title = {A Fast Quantum Mechanical Algorithm for Database Search},
    year = {1996},
    isbn = {0897917855},
    publisher = {Association for Computing Machinery},
    address = {New York, NY, USA},
    doi = {10.1145/237814.237866},
    booktitle = {Proceedings of the Twenty-Eighth Annual ACM Symposium on Theory of Computing},
    pages = {212–219},
    numpages = {8},
    venue = {Philadelphia, Pennsylvania, USA},
    series = {STOC '96},
	archivePrefix = {arXiv},
    eprint = {quant-ph/9605043}
}

@inproceedings{ABIKPV19,
    author = {Ambainis, Andris and Balodis, Kaspars and Iraids, Jānis and Kokainis, Martins and Prūsis, Krišjānis and Vihrovs, Jevgēnijs},
    title = {Quantum Speedups for Exponential-Time Dynamic Programming Algorithms},
    year = {2019},
    publisher = {Society for Industrial and Applied Mathematics},
    address = {USA},
    booktitle = {Proceedings of the Thirtieth Annual ACM-SIAM Symposium on Discrete Algorithms},
    pages = {1783–1793},
    numpages = {11},
    venue = {San Diego, California, USA},
    series = {SODA '19},
    doi = {10.1137/1.9781611975482.107},
	archivePrefix = {arXiv},
	eprint    = {1807.05209},
	primaryClass = "quant-ph"
}

@article{HK62,
	author = {Held, Michael and Karp, Richard M.},
	title = {A dynamic programming approach to sequencing problems},
	journal = {Journal of SIAM},
	volume = {10},
	number = {1},
	pages = {196-210},
	year = {1962},
	doi = {10.1145/800029.808532}
}

@article{Bel62,
	author = {Bellman, Richard},
	title = {Dynamic Programming Treatment of the Travelling Salesman Problem},
	journal = {J. ACM},
	volume = {9},
	number = {1},
	year = {1962},
	pages = {61--63},
	publisher = {ACM},
	address = {New York, NY, USA},
	doi = {10.1145/321105.321111}
}

@misc{DH96,
	author = {Dürr, Christoph and Høyer, Peter},
	title = {A Quantum Algorithm for Finding the Minimum},
	archivePrefix = {arXiv},
    eprint = {quant-ph/9607014},
	year = {1996}
}

@book{FK10,
	title={Exact Exponential Algorithms},
	author={Fomin, Fedor V. and Kratsch, Dieter},
	year={2010},
	publisher={Springer Science \& Business Media},
	isbn={978-3-642-16533-7}
}

@article{GLM08,
	title = {Quantum Random Access Memory},
	author = {Giovannetti, Vittorio and Lloyd, Seth and Maccone, Lorenzo},
	journal = {Phys. Rev. Lett.},
	volume = {100},
	issue = {16},
	pages = {160501},
	numpages = {4},
	year = {2008},
	publisher = {American Physical Society},
    doi = {10.1103/PhysRevLett.100.160501},
	archivePrefix = {arXiv},
	eprint    = {0708.1879},
	primaryClass = "quant-ph"
}

\end{document}